\documentclass[aps,pra,onecolumn,notitlepage,superscriptaddress,showpacs,nofootinbib]{revtex4-2}
\usepackage{amsmath}
\usepackage{amssymb}
\usepackage{amsfonts}
\usepackage{enumerate}
\usepackage{mathrsfs}
\usepackage{graphicx}
\usepackage{epstopdf}
\usepackage{changepage}
\usepackage{bm}
\usepackage{multirow}
\usepackage{lipsum}
\usepackage{booktabs}
\usepackage{dsfont}
\usepackage{appendix}
\usepackage{amsmath}
\usepackage{tikz}
\usepackage{multirow}
\usepackage{booktabs}
\usepackage{colortbl}
\usepackage{array}
\usepackage{diagbox}
\usepackage{utfsym}
\usetikzlibrary{matrix}
\usetikzlibrary{decorations.pathreplacing}
\tikzset{mybrace/.style={decorate, decoration={brace, mirror, amplitude=5pt}}}

\usepackage{physics} 

\newtheorem{theorem}{Theorem}
\newtheorem{definition}{Definition}
\newtheorem{observation}{Observation}
\newtheorem{lemma}{Lemma}
\newtheorem{proposition}{Proposition}
\newtheorem{conjecture}{Conjecture}
\newtheorem{example}{Example}
\newtheorem{corollary}{Corollary}

\def\squareforqed{\hbox{\rlap{$\sqcap$}$\sqcup$}}
\def\qed{\ifmmode\squareforqed\else{\unskip\nobreak\hfil
		\penalty50\hskip1em\null\nobreak\hfil\squareforqed
		\parfillskip=0pt\finalhyphendemerits=0\endgraf}\fi}
\def\endenv{\ifmmode\;\else{\unskip\nobreak\hfil
		\penalty50\hskip1em\null\nobreak\hfil\;
		\parfillskip=0pt\finalhyphendemerits=0\endgraf}\fi}
\newenvironment{proof}{\noindent \textbf{{Proof.~} }}{\qed}
\def\Dbar{\leavevmode\lower.6ex\hbox to 0pt
	{\hskip-.23ex\accent"16\hss}D}
\makeatletter
\def\url@leostyle{%
	\@ifundefined{selectfont}{\def\UrlFont{\sf}}{\def\UrlFont{\small\ttfamily}}}
\makeatother
\def\bcj{\begin{conjecture}}
	\def\ecj{\end{conjecture}}
\def\bcr{\begin{corollary}}
	\def\ecr{\end{corollary}}
\def\bd{\begin{definition}}
	\def\ed{\end{definition}}
\def\bea{\begin{eqnarray}}
	\def\eea{\end{eqnarray}}
\def\bem{\begin{enumerate}}
	\def\eem{\end{enumerate}}
\def\bex{\begin{example}}
	\def\eex{\end{example}}
\def\bim{\begin{itemize}}
	\def\eim{\end{itemize}}
\def\bl{\begin{lemma}}
	\def\el{\end{lemma}}
\def\bpf{\begin{proof}}
	\def\epf{\end{proof}}
\def\bpp{\begin{proposition}}
	\def\epp{\end{proposition}}
\def\bqu{\begin{question}}
	\def\equ{\end{question}}
\def\br{\begin{remark}}
	\def\er{\end{remark}}
\def\bt{\begin{theorem}}
	\def\et{\end{theorem}}
	\def\bt{\begin{observation}}
		\def\et{\end{observation}}

\def\btb{\begin{tabular}}
	\def\etb{\end{tabular}}

	\newcommand{\nc}{\newcommand}
	
	\nc{\bbA}{\mathbb{A}} \nc{\bbB}{\mathbb{B}} \nc{\bbC}{\mathbb{C}}
	\nc{\bbD}{\mathbb{D}} \nc{\bbE}{\mathbb{E}} \nc{\bbF}{\mathbb{F}}
	\nc{\bbG}{\mathbb{G}} \nc{\bbH}{\mathbb{H}} \nc{\bbI}{\mathbb{I}}
	\nc{\bbJ}{\mathbb{J}} \nc{\bbK}{\mathbb{K}} \nc{\bbL}{\mathbb{L}}
	\nc{\bbM}{\mathbb{M}} \nc{\bbN}{\mathbb{N}} \nc{\bbO}{\mathbb{O}}
	\nc{\bbP}{\mathbb{P}} \nc{\bbQ}{\mathbb{Q}} \nc{\bbR}{\mathbb{R}}
	\nc{\bbS}{\mathbb{S}} \nc{\bbT}{\mathbb{T}} \nc{\bbU}{\mathbb{U}}
	\nc{\bbV}{\mathbb{V}} \nc{\bbW}{\mathbb{W}} \nc{\bbX}{\mathbb{X}}
	\nc{\bbZ}{\mathbb{Z}}
	
	\nc{\bA}{{\bf A}} \nc{\bB}{{\bf B}} \nc{\bC}{{\bf C}}
	\nc{\bD}{{\bf D}} \nc{\bE}{{\bf E}} \nc{\bF}{{\bf F}}
	\nc{\bG}{{\bf G}} \nc{\bH}{{\bf H}} \nc{\bI}{{\bf I}}
	\nc{\bJ}{{\bf J}} \nc{\bK}{{\bf K}} \nc{\bL}{{\bf L}}
	\nc{\bM}{{\bf M}} \nc{\bN}{{\bf N}} \nc{\bO}{{\bf O}}
	\nc{\bP}{{\bf P}} \nc{\bQ}{{\bf Q}} \nc{\bR}{{\bf R}}
	\nc{\bS}{{\bf S}} \nc{\bT}{{\bf T}} \nc{\bU}{{\bf U}}
	\nc{\bV}{{\bf V}} \nc{\bW}{{\bf W}} \nc{\bX}{{\bf X}}
	\nc{\ba}{{\bf a}} \nc{\be}{{\bf e}} \nc{\bu}{{\bf u}}
	\nc{\brr}{{\bf r}} \nc{\bx}{{\bf x}}
	
	\nc{\cA}{{\cal A}} \nc{\cB}{{\cal B}} \nc{\cC}{{\cal C}}
	\nc{\cD}{{\cal D}} \nc{\cE}{{\cal E}} \nc{\cF}{{\cal F}}
	\nc{\cG}{{\cal G}} \nc{\cH}{{\cal H}} \nc{\cI}{{\cal I}}
	\nc{\cJ}{{\cal J}} \nc{\cK}{{\cal K}} \nc{\cL}{{\cal L}}
	\nc{\cM}{{\cal M}} \nc{\cN}{{\cal N}} \nc{\cO}{{\cal O}}
	\nc{\cP}{{\cal P}} \nc{\cQ}{{\cal Q}} \nc{\cR}{{\cal R}}
	\nc{\cS}{{\cal S}} \nc{\cT}{{\cal T}} \nc{\cU}{{\cal U}}
	\nc{\cV}{{\cal V}} \nc{\cW}{{\cal W}} \nc{\cX}{{\cal X}}
	\nc{\cZ}{{\cal Z}}
	
	\nc{\hA}{{\hat{A}}} \nc{\hB}{{\hat{B}}} \nc{\hC}{{\hat{C}}}
	\nc{\hD}{{\hat{D}}} \nc{\hE}{{\hat{E}}} \nc{\hF}{{\hat{F}}}
	\nc{\hG}{{\hat{G}}} \nc{\hH}{{\hat{H}}} \nc{\hI}{{\hat{I}}}
	\nc{\hJ}{{\hat{J}}} \nc{\hK}{{\hat{K}}} \nc{\hL}{{\hat{L}}}
	\nc{\hM}{{\hat{M}}} \nc{\hN}{{\hat{N}}} \nc{\hO}{{\hat{O}}}
	\nc{\hP}{{\hat{P}}} \nc{\hR}{{\hat{R}}} \nc{\hS}{{\hat{S}}}
	\nc{\hT}{{\hat{T}}} \nc{\hU}{{\hat{U}}} \nc{\hV}{{\hat{V}}}
	\nc{\hW}{{\hat{W}}} \nc{\hX}{{\hat{X}}} \nc{\hZ}{{\hat{Z}}}
	
	\nc{\hn}{{\hat{n}}}

	\usepackage[
	colorlinks,
	linkcolor = blue,
	citecolor = blue,
	urlcolor = blue]{hyperref}
	\def \qed {\hfill \vrule height7pt width 7pt depth 0pt}
	
	\newcounter{lastnote}

	\usepackage{xcolor}
	\definecolor{cream}{RGB}{203, 237, 204}

	\definecolor{lightyellow}{RGB}{255, 255, 204}
	\definecolor{lightblue}{RGB}{204, 204, 255}
	\definecolor{lightgreen}{RGB}{204, 255, 204}
	\definecolor{red}{RGB}{255, 102, 102}

\begin{document}
	\title{New constructions of multipartite entanglement resistant to particle loss}

\author{Wanchen Zhang}
\affiliation{School of Mathematical Sciences,
	University of Science and Technology of China, Hefei, 230026,  China}
\affiliation{Hefei National Laboratory, University of Science and Technology of China, Hefei, 230088, China}

\author{Zicheng Han }
\affiliation{School of Mathematical Sciences,
	University of Science and Technology of China, Hefei, 230026,  China}

\author{Fei Shi}
\affiliation{Department of Computer Science, School of Computing and Data Science, University of Hong Kong, Hong Kong, 999077, China}	

\author{Xiande Zhang}
\email[]{Corresponding author: drzhangx@ustc.edu.cn}
\affiliation{School of Mathematical Sciences,
	University of Science and Technology of China, Hefei, 230026,  China}
\affiliation{Hefei National Laboratory, University of Science and Technology of China, Hefei, 230088, China}

\begin{abstract}
 An entangled state is called $m$-resistant if it remains entangled after losing an arbitrary subset of $m$ particles but becomes fully separable after losing any number of particles larger than $m$. Quinta \emph{et al.} [\href{https://journals.aps.org/pra/abstract/10.1103/PhysRevA.100.062329}{Phys. Rev. A  (2019)}] conjectured that for any $N$-particle systems, there always exists  an $m$-resistant pure state.
 In this paper, we give two general constructions of $m$-resistant pure states. One is from the mixtures of Dicke states, which provides strong $(N-k)$-resistant pure $N$-qubit states with $k=4$ or $5$. The other is from classical error correcting codes, which provides new $m$-resistant qudit states for certain $m<N/2$.
\end{abstract}
\maketitle
\vspace{-0.5cm}
%

\section{Introduction}\label{sec:int}


  Entanglement is one of the most remarkable features of quantum mechanics, showcasing strong correlations between particles in a quantum system such that the state of one cannot be described independently of the others, even when separated by large distances.
   The persistence or alteration of entanglement in reduced states hinges on the initial state and the specific manner in which subsystems are lost \cite{PhysRevA.59.141,PhysRevA.65.032328,sugita2007borromean,PhysRevA.97.042307,PhysRevA.98.062335,PhysRevA.100.062329}. Various criteria and measures, including partial transposition, entanglement witnesses, and other entanglement monotones, have been developed to assess the entanglement properties of these reduced states \cite{PhysRevLett.109.233601,RevModPhys.81.865,GUHNE20091}. The investigation into how entanglement endures or transforms under the loss of subsystems remains an active area of research within quantum information theory \cite{aravind1997borromean,PhysRevA.59.141,PhysRevA.63.020303,PhysRevA.65.032328,sugita2007borromean,balasubramanian2017multi,PhysRevA.97.042307,PhysRevA.98.062335,PhysRevA.100.062329,PhysRevResearch.3.043120}.
  When some parties in a multipartite entangled pure state are lost, a natural question arises: does the reduced state remain entangled? To answer this question, previous works have referred to entangled pure states that retain their entanglement properties despite the loss of particles as being robust against particle loss \cite{PhysRevA.59.141,PhysRevA.63.020303,PhysRevA.98.062335,PhysRevResearch.3.043120}.

   For instance, consider a three-qubit \emph{Greenberger-Horne-Zeilinger} (GHZ) state
 $\ket{\text{GHZ}} = (\ket{000} + \ket{111})/\sqrt{2}$.
  Upon losing one qubit, the resulting two-qubit state will no longer exhibit entanglement but instead become a separable state devoid of quantum correlations. Consequently, GHZ states lack robustness against particle loss.
  In contrast, there exist multipartite entangled states that can preserve some form of entanglement after losing one or more parties. A prime example is the three-qubit W state, defined as: $ |\text{W}\rangle = (|100\rangle + |010\rangle + |001\rangle)/\sqrt{3}$. Even after losing one qubit, the remaining two-qubit state continues to display entanglement. This phenomenon makes W states resilient to the loss of a single particle.
  Notably, robustness to losses has significant implications from a physical perspective. For instance, Dicke states (including W states) exhibit remarkable resilience against particle loss, which makes them particularly suitable for applications such as quantum memories \cite{lvovsky2009optical,PhysRevA.86.042113}.

 In general,   an entangled pure state is called $m$-resistant if the entanglement of the reduced state of $N-m$ subsystems is fragile with respect to the loss of any additional subsystem \cite{PhysRevA.100.062329,PhysRevA.98.062335}.
  In recent years more and more researchers have started to pay attention to this kind of states with special properties. On the one hand, the existence of such states can refine the classification of multipartite entanglement \cite{PhysRevA.100.062329,PhysRevA.97.042307}. On the other hand,
  entanglement robustness and fragility has been studied for a long time \cite{PhysRevLett.80.2245,PhysRevA.62.050302,PhysRevLett.86.910,PhysRevA.63.020303,PhysRevA.65.032328,PhysRevA.67.022112,PhysRevA.86.052335,PhysRevA.87.052117}. The $m$-resistant property is a critical condition for entanglement robustness against particle loss, which is robust with respect to the loss of any $m$ particles but not for $m+1$ \cite{PhysRevA.100.062329,PhysRevA.98.062335}.

   Quinta \emph{et al.} \cite{PhysRevA.98.062335} conjectured that for any $N$-particle systems, there always exists an $m$-resistant state \cite{PhysRevA.100.062329}. In this paper,
   we give two general construction of $m$-resistant states. The main contributions are summarized as follows:
   \begin{itemize}
   	\item The concept of \emph{strong} $m$-resistant states is proposed based on the original concept of $m$-resistant states.
   	Compared to $m$-resistant states, which reflect the entanglement robustness against particle loss, strong $m$-resistant states reflect the genuine entanglement robustness against particle loss.
   	A strong $m$-resistant state is of course an $m$-resistant state, and there are more application scenarios for  strong $m$-resistant states \cite{PhysRevA.58.4394,steane1998quantum,PhysRevLett.109.233601,PhysRevResearch.3.043120,shi2024entanglement}.
   	\item We present a method to construct strong $m$-resistant pure $N$-qubit states using mixtures of Dicke states. The proposal of this method is significant for several reasons. Firstly, this method gives a construction of strong $(N-k)$-resistant states with $k=4$ or $5$, which was not obtained by any of the previous methods. Secondly, this method is applicable to arbitrary strong $m$-resistant state constructions in the sense that,  if the mixture of Dicke states is chosen delicately enough, it is promising to obtain more results that have not been obtained before. Finally, it is difficult to give a sufficiently necessary condition for a multipartite state to be separable or not, which is also the difficulty in the construction of $m$-resistant  states. But our method cleverly avoids this obstacle by restricting to the mixture of Dicke states, which makes the verification much easier.
   	\item We present a construction of $m$-resistant qudit states from error correcting  codes with large size.
   	This method serves as a complement to constructions of $m$-resistant states when $m \le N/2-1$, where previous methods often prove challenging.
   It also establishes a link between the construction of $m$-resistant states to the main coding problem.
   	As a result, any future progress on this coding problem will  lead to new $m$-resistant states.
   \end{itemize}

   The rest of this paper is organized as follows. In \autoref{sec:pre}, the preliminary knowledge and the notion of strong $m$-resistant are introduced. In \autoref{sec:Dicke}, we present a method to construct strong $m$-resistant pure states by mixtures of Dicke states.
   In \autoref{sec:code}, we construct $m$-resistant qudit states from codes with large size.
   Finally, we conclude in \autoref{sec:con}.

 \section{Preliminaries}\label{sec:pre}
 	In this section, we introduce some preliminary knowledge and facts.
 For a given set of indices $\mathcal{J}\subset [N]:=\{1, 2, \ldots, N\}$, let $\bar{\mathcal{J}}$ denote the complement of  $\mathcal{J}$ in $[N]$.

 A state is fully separable if it can be written as a convex combination
 of product states,
 \begin{equation}
 	\sum_{i}p_i\ketbra{a_i}\otimes\ketbra{b_i}\otimes \cdots \otimes \ketbra{z_i},
 \end{equation}
  where the $p_i$ are non-negative $(p_i\ge 0)$ and normalized $(\sum_{i}p_i=1)$ \cite{GUHNE2011406}. Otherwise, it is entangled.
 Consider a quantum  state $\ket{\psi}$ of $N$ particles.
    Take the partial
 trace over all particles in $\mathcal{J} \subset [N]$, and denote
 the reduced state
	\begin{equation}
	\rho_{\bar{\mathcal{J}}}(\psi)\triangleq \text{Tr}_{\mathcal{J}}(\ketbra{\psi}).
\end{equation}
An entangled state $\ket{\phi}$ is said to
be \emph{fragile} (resp.\ \emph{robust}) with respect to the loss of a given subset
$\mathcal{J}$ of the particles if $\rho_{\bar{\mathcal{J}}}(\phi)$ is fully separable (resp.\ entangled) \cite{PhysRevA.98.062335}. The concepts of entanglement robustness and fragility against particle loss are only meaningful for multipartite systems involving at least three particles. Therefore, we will assume hereafter that
$N\ge3$ for all $N$ particles under consideration.

As a  classification of the multipartite entanglement,  Quinta \emph{et al.} \cite{PhysRevA.100.062329} introduced the notion of $m$-resistant states.


\begin{definition} \cite{PhysRevA.100.062329}
	An entangled state $\ket{\phi}$ of $N$ particles is called \emph{$m$-resistant}  if it satisfies the following properties:
	\begin{itemize}
		\item [1)] $\ket{\phi}$ is robust with respect to the loss of any $m$ particles, i.e., $\rho_{\bar{\mathcal{J}}}(\phi)$ is entangled for any $\mathcal{J}\subset [N]$
		with
		$|\mathcal{J}|=m$.
		
		\item [2)] $\ket{\phi}$ is fragile with respect to the loss of any $m+1$ particles, i.e., $\rho_{\bar{\mathcal{J}}}(\phi)$ is fully separable for any $\mathcal{J}\subset [N]$ with $|\mathcal{J}|=m+1$.
	\end{itemize}
\end{definition}

The notion of $m$-resistant states  classifies the entanglement of multipartite states in terms of the difference in the topology of the entanglement \cite{PhysRevA.97.042307,PhysRevA.100.062329}. For \emph{mixed} qubit states,  it is easy to construct $m$-resistant mixed $N$-qubit states for any $m$ and $N$ \cite{PhysRevA.100.062329}. However, it is not the case for \emph{pure} qubit states. Quinta \emph{et al.} \cite{PhysRevA.100.062329} gave  constructions of $0$-, $(N-2)$-, $(N-3)$-resistant pure states of $N$ qubits. The first open case is the $1$-resistant pure states of five qubits, which they guess does not exist. They also gave constructions of $m$-resistant pure states with $N\ge 2(m+1)$ for some local dimension $d$ from a family of combinatorial designs called orthogonal arrays. They  conjectured that for any $N$ and $m$,  there  exists an $m$-resistant $N$-qudit state for some local dimension $d$, which is true for $m\leq (N-2)/2$
by the existence of orthogonal arrays.
\begin{conjecture}[\cite{PhysRevA.100.062329}]\label{conj}
  For any $N$ and $m$,  there  exists an $m$-resistant $N$-qudit state for some local dimension $d$.
\end{conjecture}

As a global form of entanglement arising in multipartite systems, \emph{genuine multipartite entanglement} (GME) is a powerful resource in many applications \cite{PhysRevA.58.4394,steane1998quantum,PhysRevLett.109.233601}.
 An $N$-particle state  $\rho$
on Hilbert space $\mathcal{H}_{1}\otimes \cdots\otimes \mathcal{H}_{N}$
is genuinely entangled if it cannot be
decomposed into as
\begin{equation}\label{eqbs}
	\rho=\sum_{S|\bar{S}}q_{S|\bar{S}}\sum_{j} p_{S|\bar{S}}^{(j)} \rho_S^{(j)}\otimes \rho_{\bar{S}}^{(j)},
\end{equation}
 where the first sum goes over all bipartitions of $[N]$.
  If a state does admit such a decomposition, it is called
 biseparable \cite{PhysRevD.35.3066,PhysRevA.100.062318,PhysRevResearch.3.043120}.

Quantum states that can maintain robust GME have also attracted a lot of attention recently \cite{PhysRevResearch.3.043120,shi2024entanglement}.
From this perspective, we give \emph{strong} $m$-resistant notions to describe the genuine entanglement resistant to particle loss.

%
%

\begin{definition}
	A genuinely entangled state $\ket{\phi}$ is called \emph{strong $m$-resistant} if it satisfies the following properties:
	\begin{itemize}
		\item [1)] $\rho_{\bar{\mathcal{J}}}(\phi)$ is genuinely entangled for any $\mathcal{J}\subset [N]$ with $|\mathcal{J}|=m$.
		
		\item [2)]  $\rho_{\bar{\mathcal{J}}}(\phi)$ is fully separable for any $\mathcal{J}\subset [N]$ with $|\mathcal{J}|=m+1$.
	\end{itemize}
\end{definition}

It is easy to observe that if an entangled state is strong $m$-resistant then it must be $m$-resistant. Indeed, the constructions of $0$-, $(N-2)$-, $(N-3)$-resistant states given by Quinta \emph{et al.} \cite{PhysRevA.100.062329} are also strong $m$-resistant. In the next section, we present a way to  find more  strong $m$-resistant entangled states, which are of course also $m$-resistant.

\section{Constructions from mixture of Dicke states}\label{sec:Dicke}

In this section, we  provide a general idea of  constructing $m$-resistant pure qubit states  by symmetric states. The reason of using symmetric states is because both entanglement and separability are required for any partial trace, and symmetric states will provide much convenience for verification.
 Dicke states are a class of symmetric states.
 Yu \cite{PhysRevA.94.060101} gave a sufficiently necessary condition for any mixture of Dicke states to be separable, which is the key in our construction.

 We denote a Dicke state as
 \begin{equation}
 	\ket{D^N_m}=\frac{1}{\sqrt{\binom{N}{m}}}\sum_{P}P\{\ket{1}^{\otimes m}\ket{0}^{\otimes (N-m)}\}
 \end{equation}
 with the sum going through all permutations $P$ over $[N]$ \cite{dicke1954coherence}.

 For any mixture of Dicke states $\rho=\sum_{i=0}^{k} b_i \ketbra{D^k_i}{D^k_i}$, we define Hankel matrices \cite{partington1988introduction} $M_0$ and $M_1$ as
 \begin{equation}
 	M_0(\rho):=\left(
 	\begin{array}{lll}
 		a_0 & \cdots & a_{m_0} \\
 		\ \vdots & \ddots & \ \vdots \\
 		a_{m_0} & \cdots & a_{2m_0} \\
 	\end{array}\right)
 	\ \text{and} \
 	M_1(\rho):=\left(
 	\begin{array}{lll}
 		a_1 & \cdots & a_{m_1} \\
 		\ \vdots & \ddots & \ \vdots \\
 		a_{m_1} & \cdots & a_{2m_1-1} \\
 	\end{array}\right),
 \end{equation}
 where $m_0 : = \lfloor \frac{k}{2} \rfloor$ and $m_1 : = \lfloor \frac{k+1}{2} \rfloor$ and $a_i:=b_i\binom{k}{i}^{-1}$.
 \begin{lemma}[\cite{PhysRevA.94.060101,partington1988introduction}]\label{sepM0M1}
 	The mixture of Dicke states $\rho=\sum_{i=0}^{k} b_i \ketbra{D^k_i}{D^k_i}$ is separable if and only if $M_0(\rho)$ and $M_1(\rho)$ are both positive semidefinite.
 \end{lemma}

In order to  construct an $m$-resistant pure state, we need to construct an entangled pure state that becomes a mixture of Dicke states after the loss of any $m$ particles, so that we can use Lemma~\ref{sepM0M1} to simultaneously determine that the reduction to $N-m-1$ particles is separable while the reduction to $N-m$ particles is entangled.

\begin{definition}\label{k-mixture}
	An $N$-qubit entangled pure state $\ket{\phi}$ is called a \emph{$k$-mixture Dicke state} if it satisfies the following properties:
	\begin{itemize}
		\item [1)]
	$\ket{\phi}$ is a linear combination of Dicke states, i.e.,  $\ket{\phi}=\sum_{i=0}^{N}c_i \ket{D^N_i}$.
		\item [2)] The reduction of $\ket{\phi}$ to any $k$ particles is a mixture of Dicke states, i.e., $\rho_{\mathcal{J}}(\phi)=\sum_{i=0}^{k} a_i \ketbra{D^k_i}{D^k_i}$ for any $|\mathcal{J}|=k$.
	\end{itemize}
\end{definition}

Note that the condition 2) of Definition~\ref{k-mixture} is satisfied if the reduction to \emph{some} $k$ particles is a mixture of Dicke states, since any $k$-reduction of a Dicke state is identical. So we just need to make sure that $\rho_{[k]}(\phi)$ is a mixture of Dicke states. If one state is a $k$-mixture Dicke state, then it is a $(k-1)$-mixture Dicke state.
Recall that symmetric states are either fully separable or GME \cite{ECKERT200288,i2016characterizing}.
Since the reduction of a Dicke state remains a symmetric state \cite{i2016characterizing}, this means if a reduction of a Dicke state is entangled, it must be GME.
Then we have the following observation by Lemma~\ref{sepM0M1}.

\begin{observation}\label{kmix} Let $N=k+m$.
	A $k$-mixture Dicke state $\ket{\phi}$ is strong $m$-resistant if and only if it satisfies the following properties:
		\begin{itemize}
		\item [1)]
		$\rho_{{[k]}}(\phi)$ is entangled, i.e., one of $M_0(\rho_{{[k]}}(\phi))$ and $M_1(\rho_{{[k]}}(\phi))$ is not positive semidefinite.
		\item [2)]
		$\rho_{{[k-1]}}(\phi)$ is separable, i.e., both
		 $M_0(\rho_{{[k-1]}}(\phi))$ and $M_1(\rho_{{[k-1]}}(\phi))$ are positive semidefinite.

	\end{itemize}
\end{observation}

In this way, we can construct a strong $m$-resistant pure state by preparing the appropriate $k$-mixture Dicke state.

Now we discuss the condition that a combination of Dicke states $\ket{\phi}=\sum_{i} c_i \ket{D^N_i}$ is a  $k$-mixture Dicke state, that is,
 any reduction of $\ket{\phi}$ to $k$ particles is a mixture of Dicke states. 
Since 
\begin{equation}
	\ketbra{\phi}{\phi}=\sum_{i,j}c_ic_j^*\ketbra{D_i^N}{D_j^N},
\end{equation}
 we have
\begin{equation}
\rho_{[k]}(\phi)=\Tr_{[N]\setminus[k]}(\ketbra{\phi}{\phi})=\sum_{i,j} c_ic_j^* \Tr_{[N]\setminus[k]} (\ketbra{D_i^N}{D_j^N})
=\sum_{i} c_ic_i^*\Tr_{[N]\setminus[k]}(\ketbra{D_i^N}{D_i^N})+\sum_{i\neq j} c_ic_j^*\Tr_{[N]\setminus[k] }(\ketbra{D_i^N}{D_j^N}).
\end{equation}

If $ c_ic_j^*\Tr_{[N]\setminus[k] }(\ketbra{D_i^N}{D_j^N})=0$ for any $i\neq j$, then
  $\rho_{[k]}(\phi)$ is a mixture of Dicke states. Thus, we have the following observation.

%

\begin{observation}\label{ob2} The combination
$\ket{\phi}=\sum_{i} c_i \ket{D^N_i}$ is a $k$-mixture Dicke state if any two nonzero coefficients $c_i$ and $c_j$  satisfy $|i-j|>k$. Specifically, we have
\begin{equation}
	\rho_{[k]}(\phi)=\Tr_{[N]\setminus[k]} (\ketbra{\phi}{\phi})=\sum_{i} c_ic_i^*\Tr_{[N]\setminus[k]} \ketbra{D^N_i}{D^N_i}.
\end{equation}
\end{observation}
\begin{proof}
 Since $|i-j|>k$, any state term $\ket{i_1i_2\cdots i_N}$ in $\ket{D^N_i}$ and any state term $\ket{j_1j_2\cdots j_N}$ in $\ket{D^N_j}$ intersect on at most $N-k-1$ particles (i.e., have the same particle values). Then tracing $N-k$ particles results a zero coefficient, that is, $ \Tr_{[N]\setminus[k] } (\ketbra{D_i^N}{D_j^N})=0$.
%
\end{proof}

In the following, we apply this method to construct strong $m$-resistant qubit states for $m=N-3,N-4$ and $N-5$.

\subsection{A construction of strong $(N-3)$-resistant qubit states}

 In \cite{PhysRevA.100.062329}, the authors proved that the state $\sqrt{\binom{N}{3}}\ket{D^N_0}+\ket{D^N_{3}}$ is $(N-3)$-resistant, whose verification is a bit complicated. In this section, we  show that all states possessing the form
\begin{equation}\label{psiN}
	\ket{\psi^{N-3}}=a\ket{D^N_0}+b\ket{D^N_{N-3}}
\end{equation}
with  $a^2\ge \frac{6(2N-5)}{N(N-1)(N-4)}b^2$
are strong $(N-3)$-resistant when $N\ge 7$.

By Observation \ref{ob2}, $\ket{\psi^{N-3}}$ is a $3$-mixture Dicke state when $N\ge 7$. Recalling that
\begin{equation}
	\rho_{[k]}(D^N_i)=\text{Tr}_{[N]\setminus[k]} \ketbra{D^N_i}{D^N_i}= \sum_{s=0}^{k} \frac{\binom{k}{s}\binom{N-k}{i-s}}{\binom{N}{i}} \ketbra{D^k_s}{D^k_s},
\end{equation}
 we have that any $2$-reduction of $\ket{\psi^{N-3}}$ can be written as

\begin{equation}
		\begin{aligned}
\rho_{[2]}(\psi^{N-3})&=
	 a^2\text{Tr}_{[N]\setminus[2]} \ketbra{D^N_0}{D^N_0}+b^2\text{Tr}_{[N]\setminus[2]} \ketbra{D^N_{N-3}}{D^N_{N-3}}
	\\&= \left(a^2+\frac{(N-2)b^2}{\binom{N}{3}}\right) \ketbra{D^2_0}+ \frac{2\binom{N-2}{2}}{\binom{N}{3}}b^2\ketbra{D^2_1} +
	\frac{\binom{N-2}{3}}{\binom{N}{3} }b^2\ketbra{D^2_2}.
	\end{aligned}
\end{equation}
%
It can be obtained that $\rho_{[2]}(\psi^{N-3})$ 
is separable since the two Hankel matrices  $M_0\propto  \left(\begin{smallmatrix}
\binom{N}{3}a^2+(N-2)b^2 & \binom{N-2}{2}b^2 \\
\binom{N-2}{2}b^2 & \binom{N-2}{3}b^2
\end{smallmatrix}\right)$ and $M_1=\left(\begin{smallmatrix}
3 \frac{N-3}{N(N-1)}b^2
\end{smallmatrix}
\right)$
are both positive semidefinite by the fact of $a^2\ge \frac{6(2N-5)}{N(N-1)(N-4)}b^2$.

Meanwhile, any $3$-reduction of $\ket{\psi^{N-3}}$ can be written as
\begin{equation}
	\begin{aligned}
\rho_{[3]}(\psi^{N-3})&=
		a^2\text{Tr}_{[N]\setminus[3]} \ketbra{D^N_0}{D^N_0}+b^2\text{Tr}_{[N]\setminus[3]} \ketbra{D^N_{N-3}}{D^N_{N-3}}
		\\&= \left(a^2+\frac{ b^2}{\binom{N}{3}}\right) \ketbra{D^3_0}+ \frac{3(N-3)}{\binom{N}{3}}b^2\ketbra{D^3_1}
		 +
		\frac{3\binom{N-3}{2}}{\binom{N}{3} }b^2\ketbra{D^3_2}+\frac{\binom{N-3}{3}}{\binom{N}{3}}b^2 \ketbra{D^3_{3}}.
%
	\end{aligned}
\end{equation}
%
Then, the Hankel matrices of any $3$-reduction of $\ket{\psi^{N-3}}$ can be written as
\begin{equation}
	M_0
\propto
	\left(
	\begin{array}{ll}
\binom{N}{3}a^2+b^2 & (N-3)b^2 \\
(N-3)b^2  &  \binom{N-3}{2}b^2 \\
	\end{array}\right)
	\ \text{and} \
	M_1\propto \left(
	\begin{array}{ll}
		N-3 &  \binom{N-3}{2} \\
		\binom{N-3}{2} &  \binom{N-3}{3} \\
	\end{array}\right).
\end{equation}
It can be checked that $M_1$ is not positive semidefinite, which means that any $3$-reduction of $\ket{\psi^{N-3}}$ is entangled.

Therefore, we can conclude that $\ket{\psi^{N-3}}$ is strong $(N-3)$-resistant for $N\geq 7$ by Observation~\ref{kmix}.

\subsection{A construction of strong $(N-4)$-resistant qubit states}
In this section, we will show that
\begin{equation}\label{psiN}
	\ket{\psi^{N-4}}=a\ket{D^N_1}+b\ket{D^N_{N-1}}
\end{equation}
 with $a, b \neq 0$ is strong $(N-4)$-resistant for $N\geq 7$. By Observation~\ref{ob2},  $\ket{\psi^{N-4}}$ is a $4$-mixture Dicke state when $N\ge 7$. Computing
\begin{equation}
		\begin{aligned}
\rho_{[3]}(\psi^{N-4})&=
	a^2\text{Tr}_{[N]\setminus[3]} \ketbra{D^N_1}{D^N_1}+b^2\text{Tr}_{[N]\setminus[3]} \ketbra{D^N_{N-1}}{D^N_{N-1}}
	\\&= a^2\left(\frac{N-3}{N} \ketbra {D^3_0}{D^3_0} + \frac{3}{N} \ketbra{D^3_1}{D^3_1})+b^2(\frac{3}{N} \ketbra {D^3_2}{D^3_2} + \frac{N-3}{N} \ketbra{D^3_3}{D^3_3}\right),
	\end{aligned}
\end{equation}
we have that the two Hankel matrices of  $\rho_{[3]}(\psi^{N-4})$ are
\begin{equation}
	M_0
	\propto
	\left(
	\begin{array}{ll}
		N-3 & 1 \\
		\	1 & 1 \\
	\end{array}\right)
	\ \text{and} \
	M_1\propto \left(
	\begin{array}{ll}
		1 &  1 \\
		1 &  N-3 \\
	\end{array}\right).
\end{equation}
By  Lemma~\autoref{sepM0M1}, we conclude that any $3$-reduction of $\ket{\psi^{N-4}}$ is separable.
Meanwhile, we have
%
\begin{equation}
\rho_{[4]}(\psi^{N-4})=a^2\left(\frac{N-4}{N} \ketbra {D^4_0}{D^4_0} + \frac{4}{N} \ketbra{D^4_1}{D^4_1})+ b^2 (\frac{4}{N}\ketbra {D^4_3}{D^4_3} + \frac{N-4}{N} \ketbra{D^4_4}{D^4_4}\right).
\end{equation}
Then, the Hankel matrices of any $4$-reduction of $\ket{\psi^{N-4}}$ can be written as
\begin{equation}
M_0
\propto
\left(
\begin{array}{lll}
	(N-4)a^2 & a^2 & 0\\
	a^2 & 0 & b^2 \\
	0 & b^2 & (N-4)b^2 \\
\end{array}\right)
\ \text{and} \
M_1\propto \left(
\begin{array}{lll}
	a^2 &  0 \\
	0 &  b^2 \\
\end{array}\right).
\end{equation}
It can be checked that $M_0$ is not positive semidefinite, which means that any $4$-reduction of $\ket{\psi^{N-4}}$ is entangled. Then, we can conclude that $\ket{\psi^{N-4}}$ is strong $(N-4)$-resistant when $N\geq 7$ by Observation~\ref{kmix}.

\vspace{0.3cm}

Note that, by using the mixture of Dicke states, we may construct more $(N-k)$-resistant states of $N$ qubits with a constant $k\geq 5$. However, when $k$ becomes big, the choice of the coefficients of the linear combination is difficult and the varication of positivities of $M_0$ and $M_1$ will be more complicated.

For example, we can also use the mixture of Dicke states of the form $\ket{\psi^{N-5}}=a\ket{D^N_1}+b\ket{D^N_{7}}+c\ket{D^N_{N}}$ for certain $a,b,c$ to construct $(N-5)$-resistant states for $N \ge 13$. See the construction and proof  in Appendix \ref{appendix1}.

Finally, we list known existence of strong $m$-resistant states of $N$ qubits for small $N$ in Table~\ref{tabled2}.  

\begin{table}[h!]
	\caption{Existence of strong $m$-resistant states of $N$ qubits. Here, LCD means the known state is a mixture of Dicke states, while subscript $a$ means the construction is from \cite{PhysRevA.100.062329}, and subscript $b$ means the construction is from  Section III. AME means the state is an absolutely maximally entangled (AME) state of six-qubit from \cite{PhysRevA.100.062329}.}\label{tabled2}
	\centering
	\begin{tabular}{c|ccccccc}
		\toprule[1pt]
		\diagbox{$m$}{$N$} & 4 & 5 & 6 & 7 & 8 & 9&10 \\[5pt]
		\midrule[0.5pt]
		0 & LCD$_a$ & LCD$_a$&  LCD$_a$ & LCD$_a$ & LCD$_a$ & LCD$_a$ & LCD$_a$\\ [5pt]
		1 & LCD$_a$ &  &  &  &  &  &\\ [5pt]
		2 & LCD$_a$ & LCD$_a$ & AME &  &  &  &\\ [5pt]
		3 &  & LCD$_a$ & LCD$_a$ & LCD$_b$ &  & & \\ [5pt]
		4 &  &  & LCD$_a$ & LCD$_a$ &  LCD$_b$ & &\\ [5pt]
		
		5 &  &  &  & LCD$_a$ & LCD$_a$ & LCD$_b$ \\ [5pt]
		
		6 &  &  &  &  & LCD$_a$ & LCD$_a$ & LCD$_b$\\
		\bottomrule[1pt]
	\end{tabular}
\end{table}

%
%

\section{Constructions from Codes}\label{sec:code}
 In this section, we establish a close connection between combinatorial arrays
and $m$-resistant pure states with high local dimension.
Using this connection we give some new constructions of $m$-resistant states. Note that these constructions are not necessarily strong $m$-resistant.


\begin{definition} An $r \times N$ array $A$ with entries from the set $S = \{0, 1, \ldots q-1 \}$ is called a \emph{$k$-critical array}, denoted by CA$(r,N,q,k)$, if
	\begin{itemize}
		\item [(i)] for any $r \times k$ sub-array, each $k$-tuple in $S^k$ appears at most once as a row;		
		\item [(ii)]  for any $r \times (k-1)$ sub-array, there exists at least one $(k-1)$-tuple in $S^{k-1}$ that appears twice as a row.
	\end{itemize}
\end{definition}

For example,  the following array is a CA$(10,7,3,3)$.
$$\begin{array}{ccccccc}
	0 & 0 & 0 & 0 & 0 & 0 &0 \\
	0 & 0 & 1 & 2 & 2 & 1 & 1 \\
	0 & 1 & 2 & 0 & 1 & 2 & 1 \\
	0 & 2 & 0 & 1 & 1 & 1 & 2 \\
	1 & 0 & 1 & 1 & 1 & 2 & 0 \\
	1 & 1 & 0 & 1 & 2 & 0 & 1 \\
	1 & 1 & 1 & 0 & 0 & 1 & 2 \\
	2 & 0 & 2 & 2 & 1 & 0 & 2 \\
	2 & 2 & 0 & 2 & 0 & 2 & 1 \\
	2 & 2 & 2 & 0 & 2 & 1 &0 \\
\end{array}
$$

By definition, in an CA$(r,N,q,k)$, we have $r\leq q^k$. The well-known orthogonal array with unitary index in combinatorial design theory \cite{bush1952orthogonal}, denoted by OA$(r,N,q,k)$, is a special $k$-critical array with $r=q^k$, such that each $k$-tuple appears exactly once in (i), and hence all $(k-1)$-tuples appears $q$ times in (ii). In \cite{PhysRevA.100.062329},  Quinta \emph{et al.} showed that if an array $A=(a_{ij})_{1\leq i\leq r;1\leq j\leq N}$ is an  OA$(r,N,q,k)$ with $r=q^k$ and $N\geq 2k$, then the state
 \begin{equation}
 	\ket{A}=\sum_{i=1}^r  \ket{a_{i1}a_{i2}\cdots a_{iN}},
 \end{equation}
 is  a $(k-1)$-resistent state.   That is, each row in the array corresponds to exactly one nonzero term in $\ket{A}$.

Now we show a more general construction in the following theorem: a $k$-critical array creates many $(k-1)$-resistent states. For example, the above CA$(10,7,3,3)$ gives a state
 \begin{equation}
	\begin{aligned}
		\ket{A} = & c_1 \ket{	0  0  0  0  0  0 0}+c_2 \ket{0  0  1  2  2  1  1 }+c_3 \ket{0  1  2  0  1  2  1} + c_4 \ket{0  2  0  1  1  1  2} +c_5 \ket{ 1  0  1  1  1  2  0} \\ & + c_6 \ket{1  1  0  1  2  0  1}
		+ c_7 \ket{1  1  1  0  0  1  2} +c_8 \ket{2  0  2  2  1  0  2} + c_9 \ket{2  2  0  2  0  2  1} + c_{10} \ket{2  2  2  0  2  1 0},
	\end{aligned}	
\end{equation} which is  $2$-resistent of $7$ particles with local dimension $3$ for any nonzero $c_i\in \mathbb{C}$. Previously in \cite{PhysRevA.100.062329}, a construction of $2$-resistant state of $7$ particles requires the local dimension at least $7$.

\begin{theorem}\label{them1}
If $A=(a_{ij})_{1\leq i\leq r;1\leq j\leq N}$ is a CA$(r,N,q,k)$ with $N\geq 2k$, then
$\ket{A}=\sum_{i=1}^r  c_i \ket{a_{i1}a_{i2}\cdots a_{iN}}$ is $(k-1)$-resistant for any $0\neq c_i\in \mathbb{C}$.
\end{theorem}
The  proof of Theorem \ref{them1} is similar to that of \cite[Proposition 1]{PhysRevA.100.062329}, and thus moved to Appendix \ref{appendix2}.

In coding theory, \emph{an  $(N,r,d)_q$ code} $\mathcal{C}$ over $\mathbb{F}_q$ is a subset of $r$ vectors in $\mathbb{F}_q^N$, called \emph{codewords}, such that any two codewords have Hamming distance at least $d$, that is, the number of coordinates they differ is at least $d$. For convenience, we arrange all codewords in an array with each row a codeword, then we get an $r\times N$ array. For example, the array CA$(10,7,3,3)$ before is indeed a $(7,10,5)_3$ code.

Now we consider the properties that maintained by the array $A$ formed by an $(N,r,d)_q$ code. In  any $r\times (N-d+1)$ subarray of $A$, there is no repeated $(N-d+1)$-tuples in $\mathbb{F}_q^{N-d+1}$, since otherwise there will be two codewords with Hamming distance less than $d$. Further, if $r>q^{N-d}$, in any $r\times (N-d)$ subarray of $A$, there  exists at least one $(N-d)$-tuple in $\mathbb{F}_q^{N-d}$ that appears twice as a row by the Pigeonhole principle. So an $(N,r,d)_q$ code with $r>q^{N-d}$ is a CA$(r,N,q,N-d+1)$.  Then we have the following corollary of \autoref{them1}.

\begin{corollary} \label{codeconditionw}
If $A=(a_{ij})_{1\leq i\leq r;1\leq j\leq N}$ is an $(N,r,d)_q$ code with $N\leq 2d-2$ and $r>q^{N-d}$, then
$\ket{A}=\sum_{i=1}^r  c_i \ket{a_{i1}\cdots a_{iN}}$ is $(N-d)$-resistant for any $0\neq c_i\in \mathbb{C}$.
%

\end{corollary}

Note that in Corollary~\ref{codeconditionw}, we need the minimum distance $d$ at least half of $N$, so  the resultant $m$-resistant state has $m=N-d<N/2$. This is a complement of the methods by using the mixture of Dicke states, which may produce $m$-resistant state with $m>N/2$.


Let $A_q(N,d)$ denote the largest possible size $r$ for which there exists an $(N,r,d)_q$ code.  Then Corollary \autoref{codeconditionw} says that if $A_q(N,d)>q^{N-d}$ and $N\leq 2d-2$, an $(N-d)$-resistant state exists for a system of $N$ particles each with a local dimension
$q$.
The problem of determining the value of $A_q(n,d)$ is one of  the main problems in coding theory \cite{roman1992coding}.
Although determining the exact value of $A_q(n,d)$ is challenging, known results in the literature \cite{vaessens1993genetic,mackenzie1984maximal,bogdanova2001bounds,laaksonen2019new} can still give us many new  $m$-resistant states. We list these results for $N\leq 10$ in Tables~\ref{tabled3}.

\begin{table}[h!]
	\caption{The existence of $m$-resistant states for $N$  particles  of local dimensions $q=3$ or $4$ from codes. The table  reads as follows. If the $(m,N)$ cell is not empty, then the $m$-resistant states for $N$ particles  exists from the code with parameters in this cell by Corollary~\ref{codeconditionw}. For example, by the $(1,5)$ entry, a $1$-resistent state for $5$ particles with local dimension $3$ can be obtained from a $(5,6,4)_3$ code (Note that a $1$-resistent state for $5$ qubits was guessed not to exist in \cite{PhysRevA.100.062329}).   All ternary codes are from \cite{mackenzie1984maximal,vaessens1993genetic} and all quaternary codes are from \cite{bogdanova2001error}.}\label{tabled3}
	\centering
	\begin{tabular}{|c|c|c|c|c|c|c|}
		\toprule[1pt]
		\diagbox{$m$}{$N$}  & 5 & 6 & 7& 8 & 9 & 10 \\[5pt]
		\midrule[0.5pt]
		  1 & $(5,6,4)_3$ &  $(6,4,5)_3$ & $(7,8,6)_4$ & $(8,5,7)_4$ & $(9,5,8)_4$  &$(10,5,9)_4$\\[5pt]	\hline
         2 &  &   & $(7,10,5)_3$ &$(8,32,6)_4$  & $(9,18,7)_4$   &\\[5pt]	\hline
         3 &  &   &  &$(8,128,5)_4$  & $(9,70,6)_4$  &\\[5pt]	
		\bottomrule[1pt]
	\end{tabular}
\end{table}

\section{Conclusion} \label{sec:con}

In this paper, we provide two methods to construct $m$-resistant states, one is from the mixture of Dicke states, which is useful for $N$-qubit states with $m$ close to $N$, and another is from error correcting codes, which is useful for $N$-qudit states with higher local dimension and $m\leq N/2-1$.

Conjecture~\ref{conj} says that there exists an $m$-resistant
$N$-qudit state for any $N$ and $m$ with some local dimension $d$. However, the existence of $m$-resistant states for lower local dimensions implies that for higher local dimensions. So the determination of  the smallest local dimension for which there exists an $m$-resistant
$N$-qudit state is important and deserves further study. Here, we combine all known results and list the upper bound of this number in Table~\ref{tableall}. Previously, $1$-resistant states for systems with $5$, $6$, $7$, and $8$ particles exist for local dimensions at least $4$, $5$, $7$, and $7$, respectively \cite{PhysRevA.100.062329}.

\begin{table}[h]
	\centering
	\caption{Existence of $m$-resistant states with known minimum dimension. The number $2$ represents the lowest dimension. Dimensions greater than $2$ might be improved in the future.  All italic numbers are new from this paper. 
}\label{tableall}
	\begin{tabular}{c|c|c|c|c|c|c|c|}
		\toprule[1pt]
		\diagbox{$m$}{$N$} & 4 & 5 & 6 & 7 & 8  &9&10\\
		\midrule[0.5pt]
		0 & 2  & 2 &  2  & 2  & 2  & 2  & 2 \\
		
		1 & 2 & \emph{3}   & \emph{3}  & \emph{4} & \emph{4}  & \emph{4}&  \emph{4}\\
		
		2 & 2  & 2  & 2  & \emph{3}  & \emph{4}  & \emph{4}&  \\
		
		3 &  & 2  & 2  & \emph{2}  & \emph{4}  & \emph{4}&  \\
		
		4 &  &  & 2  & 2  & \emph{2}  & &\\
		
		5 &  &  &  & 2 & 2  & \emph{2} & \\
		
		6 &  &  &  &  & 2  & 2 & \emph{2}\\
		\bottomrule[1pt]
	\end{tabular}
\end{table}

Finally, the concept of strong $m$-resistant states was proposed to reflect the genuine entanglement robustness against particle loss.
Notice that if a state is not a GME, it is not necessarily separable, which is required in the concept of strong $m$-resistant states. So it is natural to raise the following concept: a genuinely entangled state is a  \emph{genuinely} $m$-resistant state if it  maintains GME after the loss of any $m$ particles, but becomes biseparable after the loss of any $m+1$ particles.  Then there is a relationship between the three concepts as in \autoref{fig1}. Just as states that are biseparable but not fully separable have attracted a lot of attention \cite{PhysRevLett.82.5385,PhysRevA.85.022322,RevModPhys.81.865}, it is worthwhile to explore whether there are states that are genuinely $m$-resistant but not strong $m$-resistant.

\begin{figure}
	\centering
	\includegraphics[width=0.4 \textwidth]{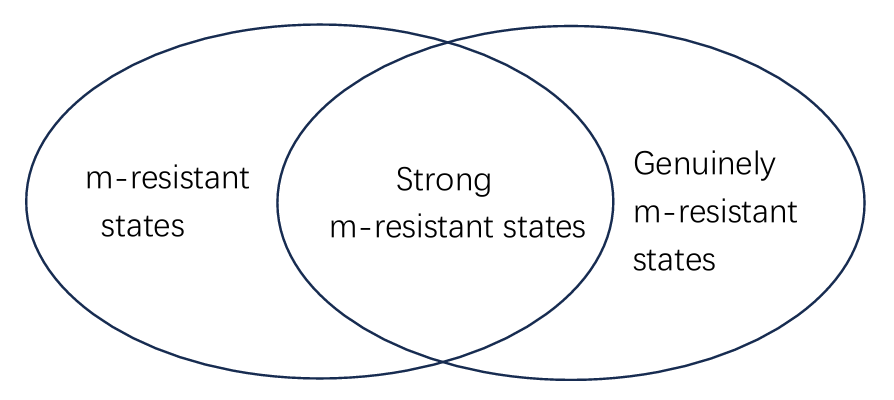}
	\caption{Strong $m$-resistant states are the intersection of $m$-resistant states and genuinely $m$-resistant states.}\label{fig1}
\end{figure}

	\section{Acknowledgments}
The research of W. Zhang, Z. Han and X. Zhang were supported by the
Innovation Program for Quantum Science and Technology 2021ZD0302902,
the NSFC under Grants No. 12171452 and No. 12231014, and the National
Key Research and Development Programs of China 2023YFA1010200 and
2020YFA0713100.
The research of F. Shi was supported by the HKU Seed Fund
for Basic Research for New Staff via Project 2201100596, Guangdong Natural
Science Fund via Project 2023A1515012185, National Natural Science Foun
dation of China (NSFC) via Project No. 12305030 and No. 12347104, Hong
Kong Research Grant Council (RGC) via No. 27300823, N\_HKU718/23,
and R6010-23, Guangdong Provincial Quantum Science Strategic Initiative
GDZX2200001.

\appendix
\section{A construction of $(N-5)$-resistant qubit states}\label{appendix1}
In this appendix, we will show that
\begin{equation}\label{psiN}
	\ket{\psi^{N-5}}=a\ket{D^N_1}+b\ket{D^N_{7}}+c\ket{D^N_{N}}
\end{equation}
with $a^2=\frac{\binom{N-4}{5}N(N-3)}{30\binom{N}{7}}b^2$ and $c$ being large compared to $a$ and $b$ is strong $(N-5)$-resistant when $N\ge 13$. By Observation~\ref{ob2},  $\ket{\psi^{N-5}}$ is a $5$-mixture Dicke state when $N\ge 13$.

Computing
\begin{equation}
	\begin{aligned}
		\rho_{[4]}(\psi^{N-5})=&
	a^2\text{Tr}_{[N]\setminus[4]} \ketbra{D^N_1}+b^2\text{Tr}_{[N]\setminus[4]} \ketbra{D^N_{7}} + c^2\text{Tr}_{[N]\setminus[4]} \ketbra{D^N_{N}}
\\=&
		a^2\left(\frac{N-4}{N} \ketbra {D^4_0} + \frac{4}{N} \ketbra{D^4_1}\right)
		 +	\\&b^2\left(\frac{\binom{N-4}{7}}{\binom{N}{7}} \ketbra {D^4_0} + 4\frac{\binom{N-4}{6}}{\binom{N}{7}} \ketbra{D^4_1}+
		6\frac{\binom{N-4}{5}}{\binom{N}{7}} \ketbra{D^4_2} +
		4\frac{\binom{N-4}{4}}{\binom{N}{7}} \ketbra{D^4_3} +
		\frac{\binom{N-4}{3}}{\binom{N}{7}} \ketbra{D^4_4}\right)\\&
		+c^2 \ketbra{D^4_4},
	\end{aligned}
\end{equation}
 we have the $M_0$ and $M_1$ of $\rho_{[4]}(\psi^{N-5})$ can be written as
\begin{equation}
	M_0
	=
	\left(
	\begin{array}{lll}
		\frac{N-4}{N}a^2 + \frac{\binom{N-4}{7}}{\binom{N}{7}}b^2 & \frac{1}{N}a^2+\frac{\binom{N-4}{6}}{\binom{N}{7}}b^2 & \frac{\binom{N-4}{5}}{\binom{N}{7}}b^2\\
		\frac{1}{N}a^2+\frac{\binom{N-4}{6}}{\binom{N}{7}}b^2 & \frac{\binom{N-4}{5}}{\binom{N}{7}}b^2 & \frac{\binom{N-4}{4}}{\binom{N}{7}}b^2\\
		\frac{\binom{N-4}{5}}{\binom{N}{7}}b^2 & \frac{\binom{N-4}{4}}{\binom{N}{7}}b^2 & \frac{\binom{N-4}{3}}{\binom{N}{7}}b^2 + c^2\\
	\end{array}\right)
	\ \text{and} \
	M_1 = \left(
	\begin{array}{lll}
		\frac{1}{N}a^2+\frac{\binom{N-4}{6}}{\binom{N}{7}}b^2 &  \frac{\binom{N-4}{5}}{\binom{N}{7}}b^2 \\
		\frac{\binom{N-4}{5}}{\binom{N}{7}}b^2 &  \frac{\binom{N-4}{4}}{\binom{N}{7}}b^2 \\
	\end{array}\right).
\end{equation}

Suppose that $a^2=\frac{\binom{N-4}{5}N(N-3)}{30\binom{N}{7}}b^2$,
then the determinant of $M_1$ is $0$, which implies that $M_1$ is positive semidefinite. It can be checked that there is a $2 \times 2$ principal submatrix in the upper left corner of $M_0$ that is positive semidefinite.
Thus we can always find $c$ large enough to make the determinant of $M_0$ greater than $0$. Then we can obtain that all leading principal minors of $M_0$ are greater than $0$, that is, $M_0$ is positive definite. Thus, $\rho_{[4]}(\psi^{N-5})$ is separable.

Computing
\begin{equation}
	\begin{aligned}
		\rho_{[5]}(\psi^{N-5})=&
		a^2\text{Tr}_{[N]\setminus[5]} \ketbra{D^N_1}+b^2\text{Tr}_{[N]\setminus[5]} \ketbra{D^N_{7}} + c^2\text{Tr}_{[N]\setminus[5]} \ketbra{D^N_{N}}
		\\=&
		a^2\left(\frac{N-5}{N} \ketbra {D^5_0} + \frac{5}{N} \ketbra{D^5_1}\right)
	 +	\\&b^2\left(\frac{\binom{N-5}{7}}{\binom{N}{7}} \ketbra {D^5_0} + \frac{5\binom{N-5}{6}}{\binom{N}{7}} \ketbra{D^5_1}+
		\frac{10\binom{N-5}{5}}{\binom{N}{7}} \ketbra{D^5_2} +
		\frac{10\binom{N-5}{4}}{\binom{N}{7}} \ketbra{D^5_3}
+\right.\\&
		\left.\frac{5\binom{N-5}{3}}{\binom{N}{7}} \ketbra{D^5_4} +
		\frac{\binom{N-5}{2}}{\binom{N}{7}} \ketbra{D^5_5}\right)
		+c^2 \ketbra{D^5_5},
	\end{aligned}
\end{equation}
 we have the $M_0$ of $\rho_{[5]}(\psi^{N-5})$ can be written as

\begin{equation}
	M_0
	\propto
	\left(
	\begin{array}{lll}
		\frac{N-5}{N}a^2 + \frac{\binom{N-5}{7}}{\binom{N}{7}}b^2 & \frac{1}{N}a^2+\frac{\binom{N-5}{6}}{\binom{N}{7}}b^2 & \frac{\binom{N-5}{5}}{\binom{N}{7}}b^2\\
		\frac{1}{N}a^2+\frac{\binom{N-5}{6}}{\binom{N}{7}}b^2 & \frac{\binom{N-5}{5}}{\binom{N}{7}}b^2 & \frac{\binom{N-5}{4}}{\binom{N}{7}}b^2\\
		\frac{\binom{N-5}{5}}{\binom{N}{7}}b^2 & \frac{\binom{N-5}{4}}{\binom{N}{7}}b^2 & \frac{\binom{N-5}{3}}{\binom{N}{7}}b^2 \\
	\end{array}\right).
\end{equation}

It is easy to see that there is a $2 \times 2$ principal submatrix in the lower right corner of $M_0$ that is not positive semidefinite, which means that $M_0$ is not positive semidefinite. This shows that $	\rho_{[5]}(\psi^{N-5})$ is entangled. Then, we can conclude that $\ket{\psi^{N-5}}$ is strong $(N-4)$-resistant when $N\geq 13$ by Observation~\ref{kmix}.

\section{Proof of Theorem 1}\label{appendix2}
%
\begin{proof}
	In order to prove that $\ket{A}$ is  $(k-1)$-resistant, we need to prove the following two properties:
	\begin{itemize}
		\item  $\rho_{\bar{\mathcal{J}}}(A)$ is entangled for any $\mathcal{J}$ with $|\mathcal{J}|=k-1$.		
		\item   $\rho_{\bar{\mathcal{J}}}(A)$ is separable for any $\mathcal{J}$ with $|\mathcal{J}|=k$.
	\end{itemize}
	By the fact that each $k$-tuple appearing at most once in any $k$ columns of $A$, we have that  $\rho_{\bar{\mathcal{J}}}(A)$ is a diagonal matrix. Hence, $\rho_{\bar{\mathcal{J}}}(A)$ is fully separable for any $\mathcal{J}$ with $|\mathcal{J}|=k$.
	
	Now, we  prove the first property, that is, $\rho_{\bar{\mathcal{J}}}(A)$ is entangled for any $\mathcal{J}$ with $|\mathcal{J}|=k-1$. We prove this result using the PPT criterion, that is, after the partial transpose of any particle, this state always has a negative eigenvalue.
     Since there exists at least one $(k-1)$-tuple that appears twice for any $k-1$ columns in $A$, we can take two rows of $A$, which are the same in the first $(k-1)$ positions.
     Without loss of generality assume that they are of the form
    	$$
    \begin{matrix}
    	& {0 \ldots 0 }& { 0 \ 0 \ldots 0} \\
    	&	\underbrace{0 \ldots 0}_{k-1} & \underbrace{i_k \ldots i_N}_{N-k+1}
    \end{matrix} ,
$$  where $i_t\neq 0$ with $k\le t \le N$.

Then the proof is almost the same as in \cite[Proposition 1]{PhysRevA.100.062329}. We sketch the proof here for completeness. 
Consider the density matrix $\rho_{[N]\setminus[k-1]}(A)$ obtained by tracing
out the first $(k-1)$ particles. Then, we examine the eigenvalues
of the partial transpose of the last particle  $\rho_{[N]\setminus[k-1]}^{T_l}(A)$.
The aforementioned rows generate the following elements in $\rho_{[N]\setminus[k-1]}^{T_l}(A)$:
$$\ketbra{0\ldots 0  i_N}{i_k \ldots i_{N-1}  0} \ \text{and} \ \ketbra{i_k \ldots i_{N-1}  0}{0\ldots 0  i_N}.$$
We claim that these elements are the only ones interlocking the bras $\bra{i_k \ldots i_{N-1}  0}$, $\bra{0\ldots 0 i_N}$ and the kets $\ket{0\ldots 0  i_N}$, $\ket{i_k \ldots i_{N-1} 0}$.
Taking $\ket{0\ldots 0  i_N}$ as an example, one can deduce it from the following two observations:
\begin{itemize}
	\item  The row $0 \ldots 0 j_k \ldots j_{N-1} i_N$ with $j_k\ldots j_{N-1} \neq i_k \ldots i_{N-1}$ will not be in $A$, otherwise there will be a repeated $k$-tuple, which is a contradiction.
	\item Since $N\ge 2k$, that is $N-k \ge k$, the row $j_1\ldots j_{k-1} \underbrace{0 \ldots 0}_{N-k} j_N$ with $j_1\ldots j_{k-1} \neq 0 \ldots 0$ will not be in $A$ either.
\end{itemize}

	Then $\rho_{[N]\setminus[k-1]}^{T_l}(A)$ can be written as
	\begin{equation}
		\rho_{[N]\setminus[k-1]}^{T_l}(\phi)\propto
		\left(\begin{array}{ccccc}
			0 & 1 & 0 \cdots & 0\\
			1 & 0 & 0 \cdots & 0\\
			0 & 0                      \\
			\vdots & \vdots & {\Huge B}   \\
			0 & 0
		\end{array}\right),
	\end{equation}
	where we have changed the order of the computational basis
	in such a way that the first two vectors are $\ket{0\ldots 0}$ and $\ket{i_k\ldots i_N}$.
	
	Consequently, $\rho_{[N]\setminus[k-1]}^{T_l}(A)$ has at least one negative eigenvalue, namely, $-1$. 
	This ends the proof.			
		\end{proof}

\bibliographystyle{IEEEtran}
\bibliography{reference}
\end{document}